\documentclass[11pt,tgrind]{article}

\usepackage[centertags]{amsmath}
\usepackage{amsthm}
\usepackage{psfrag}
\usepackage{amssymb}
\usepackage{amscd}
\usepackage{eucal}
\usepackage{epsfig}
\usepackage{verbatim}
\usepackage{color}

\newcommand{\cE}{\mathcal{E}}
\newcommand{\cO}{\mathcal{O}}

\newcommand{\beq}{\begin{eqnarray}}
\newcommand{\eeq}{\end{eqnarray}}
\newcommand{\beqn}{\begin{equation}}
\newcommand{\eeqn}{\end{equation}}
\newcommand{\R}{\mathbb{R}}

\newcommand{\beps}{\varepsilon}

\newcommand{\bone}{\mathbf{1}}

\newcommand{\mc}{\mathcal}
\newcommand{\mb}{\mathbf}

\renewcommand{\hat}{\widehat}
\newcommand{\cM}{\mc{M}}

\newcommand{\cI}{\mc{I}}

\newcommand{\cP}{\mc{P}}
\renewcommand{\P}{\mathbb P}

\newcommand{\bm}{\mb{m}}

\newcommand{\bb}{\mb{b}}
\newcommand{\bx}{\mb{x}}

\newcommand{\by}{\mb{y}}
\newcommand{\bz}{\mb{z}}

\newcommand{\E}{\mathbb{E}}

\newcommand{\hash}{h}

\newcommand{\ith}{$i^\text{th}$}
\newcommand{\jth}{$j^\text{th}$}

\newcommand{\lth}{$\ell^\text{th}$}

\newcommand{\snr}{{\sf SNR}}
\newcommand{\CAP}{C_{\sf awgn}}
\newcommand{\MapN}{c}

\newtheorem{theorem}{Theorem}
\newtheorem{lemma}{Lemma}

\newtheorem{proposition}[lemma]{Proposition}
\newtheorem{claim}[lemma]{Claim}

\newtheorem{definition}{Definition}
\newtheorem{property}[lemma]{Property}

\title{De-randomizing Shannon: The Design and Analysis of a Capacity-Achieving
Rateless Code}

\author{
 {\sf Hari Balakrishnan, Peter Iannucci, Jonathan Perry, Devavrat Shah}\\
{\em Department of EECS}\thanks{HB, PI, and JP are affiliated with CSAIL; DS
is affiliated with LIDS.}\\ {\em Massachusetts Institute of
Technology}\\{\em Cambridge, MA, USA.}
}

\date{}

\begin{document}

\maketitle

\thispagestyle{empty}
\begin{abstract}

  This paper presents an analysis of spinal codes, a class of rateless
  codes proposed recently~\cite{hotnets-spinal}.  We
  prove that spinal codes achieve Shannon capacity for the binary
  symmetric channel (BSC) and the additive white Gaussian noise (AWGN)
  channel with an efficient polynomial-time encoder and decoder.  They
  are the first rateless codes with proofs of these properties for
  BSC and AWGN.
  
  The key idea in the spinal code is the sequential application of a
  hash function over the message bits. The
  sequential structure of the code turns out to be crucial for
  efficient decoding.  Moreover, counter to the wisdom of having an
  expander structure in good codes \cite{sipser1996expander}, we show
  that the spinal code, despite its sequential structure,
  achieves capacity.  The pseudo-randomness provided by a hash function suffices for this
  purpose.
  
  Our proof introduces a variant of Gallager's result characterizing
  the error exponent of random codes for any memoryless channel
  \cite[Chapters 5, 7]{Gallager68}.  We present a novel application of
  these error-exponent results within the framework of an efficient
  sequential code.
  The application of a hash function over the message bits provides a
  methodical and effective way to de-randomize Shannon's random
  codebook construction~\cite{shannon1949communication}.  


\end{abstract}

\newpage
\setcounter{page}{1}
\section{Introduction}

In a {\em rateless code}, the codewords (i.e., coded bits or symbols)
corresponding to higher-rate encodings are prefixes of lower-rate
encodings.  Rateless codes have been known since Shannon's random
codebook construction~\cite{shannon1949communication}, which proved
the existence of capacity-achieving codes.  Unfortunately, the random
codebook is computationally intractable to decode, taking
time exponential in the message size. It took several decades of
research on coding theory and algorithms before practical rateless
codes were discovered for the binary erasure channel (BEC) by Luby (LT
codes~\cite{luby2003lt}) and Shokrollahi (Raptor
codes~\cite{shokrollahi2006raptor}). The BEC is a good model for
packet losses on the Internet.

For wireless channels, however, packet erasure models give way to more
appropriate random bit-flip models (at the link layer) and additive
noise models (at the physical layer).  Moreover, wireless channel
conditions vary with time due to mobility and interference, even over
durations as short as a single packet transmission.  In this setting,
{\em fixed-rate} (or fixed-length) codes that work well at a fixed
(and known) bit-flip probability or signal-to-noise ratio (SNR) are by
themselves insufficient to achieve high throughput; they require
additional (and complex) heuristics to determine what the channel
conditions are, and to pick the right
code~\cite{bicket-thesis,camp-mobicom08,charm-mobisys08,vutukuru2009cross-layer},
resulting in a system without any theoretically appealing properties.
This task becomes difficult with rapid channel variations, numerous
transmission rate alternatives, and multiple transmitters contending
for the same wireless channel.

In contrast to fixed-rate codes, a good rateless code will adapt
automatically to changing conditions because it will {\em inherently}
transmit just the right amount, whatever the conditions.  Because they
are a natural fit for time-varying wireless networks, the design of
good rateless codes for the {\em binary symmetric channel} (BSC) and
the {\em additive white Gaussian noise} (AWGN) channel has received
renewed interest recently~\cite{Erez12,GK11,hotnets-spinal}.  By
``good'', we mean a code that achieves a rate close to channel
capacity: $1 - H(p)$ for the BSC, where $p$ is the bit-flip
probability and $H(p) = -p \log p - (1-p) \log (1-p)$, and
$\frac{1}{2} \cdot \log (1 + \snr)$ for the AWGN channel, where $\snr$
is the ratio of the signal power to the noise variance.\footnote{In
  this paper $\log$ means ``logarithm to base 2'' and $\ln$ stands for
  the natural logarithm.}



In this paper, we prove that a family of rateless codes, called {\em
  spinal codes}, achieves capacity over both the BSC and the AWGN
channel. Spinal codes are the {\em first} provably capacity-achieving
rateless codes with a polynomial-time encoder and decoder over both
these standard channel models.  Our work provides for the BSC and AWGN
channel what LT~\cite{luby2003lt} and
Raptor~\cite{shokrollahi2006raptor} codes provide for the BEC, but
with a rather different approach.

Spinal codes use hash functions satisfying the {\em pair-wise 
indepence}~\cite{Mitz} to produce a sufficiently random codebook.  The encoder for a spinal code applies the hash
function sequentially over groups of message bits in a structure that
resembles a classic convolutional code.  The maximum-likelihood (ML)
decoder for a spinal code constructs a tree of possibilities by
replaying the encoder over various possible input message bits, and
computes either the Hamming distance (BSC) or squared Euclidean
distance (AWGN) between the received data and the various choices in
the tree.  A complete tree is, of course, exponential in the message
size, but our key result is that one can aggressively prune the
decoding tree to obtain an efficient decoder with polynomial
computational cost, that still essentially achieves capacity.

Our approach highlights how the SAC property of the hash function
provides a way to de-randomize
Shannon's
random codebook~\cite{shannon1949communication} approach to produce a
practical, capacity-achieving rateless code.  As such, our proof
methods are likely to extend to de-randomize, and possibly render
practical, various random coding constructions in Information Theory
that have hitherto been widely used to characterize {\em existential}
capacity results (cf. El Gamal and Kim~\cite{ElGamal-Kim-12}).



\paragraph{Prior work.} 
Raptor codes, though designed primarily for erasure channels (on which
they provably achieve capacity), can be extended to AWGN and BSC
channels with a belief propagation decoder~\cite{palanki2005rateless}
similar to graphical codes like
LDPC~\cite{gallager2002ldpc,vila_casado_informed_2007}.  However, not
much is known theoretically about how good this code is over these
channels. In fact, the capacity of LDPC codes (with an efficient
decoder) over both the BSC and AWGN channels, in general, is still
unresolved, which is further evidence that the BSC and AWGN channels
are non-trivial settings for the design and analysis of good codes.

Recently, an interesting ``layered'' approach has been developed by
Erez, Trott, and Wornell \cite{Erez12} (\cite{GK11} describes an
implementation of this concept) primarily for the AWGN channel, but
there is no obvious way to extend it to the BSC. In this approach, a
layered rateless code is built upon a capacity-achieving fixed-rate
``base'' code at the lowest layer.  Erez et al. prove that their code
achieves capacity over AWGN assuming that the base code achieves
capacity at some SNR and the number of layers increases without bound.
Our work is an improvement over this layered approach in two ways:
first, we resolve an open question they raise about designing an
efficient capacity-achieving rateless code for the BSC, and second, it
is a more direct and natural construction that does not rely on
layering atop a (presumed capacity-achieving) fixed-rate base code.

Structurally, spinal codes are similar to convolutional
codes~\cite{Elias55,viterbi1967error}, which apply a linear function
sequentially over the message bits, but such codes with so-called
small state (constraint length) are far from capacity (in large part
because of their sequential nature).  In contrast, to achieve capacity
using linear codes (whether fixed-rate or rateless) over the BEC,
prior work suggests that some form of random graph ensemble or
expander structure is
necessary~\cite{gallager2002ldpc,sipser1996expander}.  Somewhat
surprisingly, despite their sequential nature, we are able to
establish that spinal codes---by using a hash function with the pairwise
independence---achieve capacity.

\paragraph{Our results.} For the BSC, we show that rateless spinal
codes can be encoded in $O(\frac{n \log n}{\beps^2})$ time and decoded
in $n^{O(1/\beps^3)}$ time, where $n$ is the number of message bits
and $\beps$ is the {\em gap to capacity} at which the code is
operating (i.e., the achieved rate is within $\beps$ of capacity).
This result holds for $n = \Omega(1/\beps^5)$. For the AWGN channel,
we establish a similar result with somewhat lower computational cost:
$O(\frac{n \log n}{\beps})$ time for encoding and $n^{O(1/\beps^2)}$
time for decoding.

Thus, by selecting $n = {\sf poly}(1/\beps)$, it is possible to
operate within $\beps$ of capacity with an encoding cost ${\sf
  poly}(1/\beps)$ and decoding cost $\exp\big({\sf
  poly}(1/\beps)\big)$ for both channel models.  These costs are
comparable to the computational efficiency achieved by the Forney's
concatenation construction~\cite{Forney66}, as described in
Guruswamy's survey of iterative decoding methods~\cite{Gsurvey06}
($n/\beps^{O(1)}$ for encoding and $n2^{1/\beps^{O(1)}}$ for
decoding). However, the key advantage of spinal codes
is that they are rateless, unlike all known good and efficient codes
for the BSC; and they are arguably more elegant than the concatenation
construction.
We have implemented spinal codes in both software and hardware (FPGA)
to demonstrate their practicality and high throughput, allowing us to
project that a silicon implementation of the design will run at 50
Mbits/s (commercial 802.11b/g speeds)~\cite{sigcomm-spinal}.  The
experimental results should alleviate concerns about the
super-linearity of the encoder and decoder being a barrier to their
practical usefulness.





\paragraph{Method and proof technique.} 
The key idea is to use the error exponents of random codes as a
building block.  We apply Gallager's result characterizing the random
coding error exponent for any memoryless channel \cite[Chapters 5,
7]{Gallager68}. That result, though established for random codes where
the codewords for distinct messages are mutually independent, applies
even if only pairwise independence between the coded bits holds. The
application of this idea to analyze spinal codes is somewhat
remarkable because many coded bits of two distinct messages are likely
to be highly dependent.  The rest of the proof uses probabilistic
analysis leveraging the SAC property of the hash function as a
de-randomization strategy, establishing that a sequentially structured
code can achieve capacity.

\section{Overview of Spinal Codes}
\label{sec:one}

This section describes the encoder (\S\ref{ss:bscenc}) and decoder
(\S\ref{ss:bscdec}) for spinal codes, which are variants of the
methods introduced in~\cite{hotnets-spinal}. Our discussion here is in
the context of the BSC, but the same approach with one addition
(direct coding to symbols) works for the AWGN channel, as described in
\S\ref{s:awgn}.


\subsection{Encoder}
\label{ss:bscenc}

The encoder maps $n$ input message bits, $\bm = (m_1,\dots, m_n)$ to a
stream of coded bits, $x_1(\bm), x_2(\bm), \dots$. These coded bits
are transmitted in sequence until the receiver signals that it is done
decoding. 


\paragraph{Hash function.} The core of the code is a hash function,
$\hash$, which takes two inputs, a $\nu$-bit state and $k$ message
bits, and maps them to a new $\nu$-bit state. That is, $\hash :
\{0,1\}^\nu \times \{0,1\}^k \to \{0,1\}^\nu.$ We choose $\hash$
uniformly at random, based on a random seed, from ${\cal H}$, a family
of hash functions with {\em pair-wise} independence property cf. \cite{Mitz}: 
each $x \in \{0,1\}^\nu$ is mapped uniformly at random (randomness induced
by selection of random seed) to any of the $\{0,1\}^k$; for any $x \neq x' \in \{0,1\}^\nu$, 
\begin{align}
\P(h(x) = y, h(x') = y') & ~=~\P(h(x) = y)\, \P(h(x') = y') ~=~ 2^{-2k},
\end{align}
for any $y, y' \in \{0,1\}^k$. 

\paragraph{Spine.} $\hash$ is applied sequentially to $k$
non-overlapping message bits at a time, producing a sequence of
$\nu$-bit states called the {\em spine}.  The initial state $s_0 =
0^\nu$. Let $\bar{m}_i = (m_{ki+1}\ldots m_{k(i+1)})$ be the \ith{}
$k$-bit block of the message $\bm$. Then, as shown in
Figure~\ref{fig:encoder}, each successive $\nu$-bit value in the spine
is generated as
$$s_i = \hash(s_{i-1}, \bar{m}_{i-1}), \quad 1\leq i\leq n/k.$$

\begin{figure}[t]
\centering
 \includegraphics[width=0.9\textwidth]{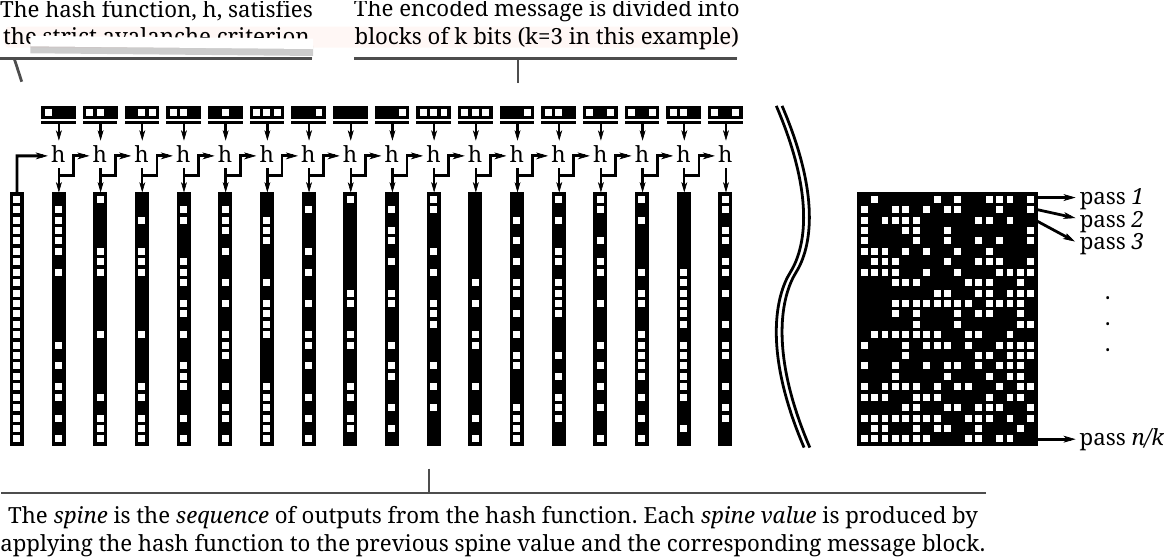}
 \caption{Encoder for the BSC. Each dark square is a ``1'', each white
   square is a ``0''.}
\label{fig:encoder}
\end{figure}

\paragraph{Generating coded bits.} The encoder uses the spine values
$s_1,\dots, s_{n/k}$ to produce coded bits in passes. In the first
pass, it extracts the most significant bit from each $\nu$-bit spine
value to produce $n/k$ coded bits $x_1,\dots, x_{n/k}$. In general, in
the \lth{} pass, the encoder extracts the \lth{} most significant bit
of each spine value $s_1,\dots, s_{n/k}$, producing coded bits
$x_{(\ell-1) \frac{n}{k} +1},\dots, x_{\ell \frac{n}{k}}$.
The coding parameters $k, \nu$ determine various properties of the
code.  The maximum rate achieved by the code at the end of the \lth{}
pass is $R_\ell = k/\ell$; the lowest achievable rate is $k/\nu$.

\paragraph{Sequential structure of the code.}
The combination of the encoder's iterative structure and the SAC
property of the hash function gives the code a unique balance. On the
one hand, two messages that differ by one or more bits will have very
different codewords, allowing analysis using random coding techniques.
On the other hand, this divergence in the output is structured in such
a way as to to allow an efficient decoder.

In a spinal code, the output bits $x_i, x_{i+\frac{n}{k}},
x_{i+2\frac{n}{k}}, \ldots$ are fully determined by the first $i \cdot
k$ bits of the message $\bm$.  Two messages that first differ in the
$i^\text{th}$ block of $k$ bits have the same first $i-1$ spine
values, and have statistically independent subsequent spine values
(i.e., the later values are ``very different'').

\subsection{Decoder}
\label{ss:bscdec}



\paragraph{Decoding over a tree.} Maximum likelihood (ML) decoding
over the BSC boils down to a search for the encoded message whose
Hamming distance is nearest to the received message.  Because the
spinal encoder applies the hash function sequentially, input messages
with a common prefix will also have a common spine prefix.  The key to
exploiting this structure is to decompose the total distance into a
sum over spine values.  If we break the received bits $\by$ into
sub-vectors $\by_1,\ldots,\by_{n/k}$ containing symbols from spine
values $s_1,\dots, s_{n/k}$ of the correct message, and similarly if
we break $\bx(\bm')$ for the candidate message $\bm'$ into $n/k$
vectors of bits $\bx_1(s'_1), \dots, \bx_{n/k}(s'_{n/k})$ that depend
on spine values $s'_1,\dots, s'_{n/k}$ (corresponding to message
$\bm'$), then the cost function decomposes as
\begin{align}
d_H(\by, \bx(\bm')) & = \sum_{i=1}^{n/k} d_H(\by_i, \bx_i(s'_i)). 
\label{eq:decompose}
\end{align}
A summand $d_H(\by_i, \bx_i(s_i))$ only needs to be computed once for
all messages that share the same spine value $s_i$. The following
algorithm takes advantage of this property.

Ignoring hash function collisions (as established in the proof of
Theorem \ref{thm:bsc}, this happens with very low probability),
decoding can be recast as a search over a tree of message prefixes.
The root of this {\em decoding tree} is $s_0$, and corresponds to the
zero-length message.  Each node at depth $d$ corresponds to a prefix
of length $kd$ bits, and is labeled with the final spine value $s_d$
of that prefix.  Every node has $2^k$ children, connected by edges
$e=(s_d,s_{d+1})$ representing a choice of $k$ message bits $\bar
m_e$.  As in the encoder, $s_{d+1}$ is $\hash(s_d, \bar m_e)$.  By
walking back up the tree to the root and reading $k$ bits from each
edge, we can find the message prefix for a given node.

To the edge incident on node $s_d$, we assign a {\em branch cost}
$d_H(\by_d, \bx_d(s_d))$.  Summing the branch costs on the path from
the root to a node gives the {\em path cost} of that node, equivalent
to the sum in Eq.\eqref{eq:decompose}.  The ML decoder finds the leaf
with the lowest cost, and returns the corresponding complete message.
The sender continues to send successive passes until the receiver
signals that the message has been decoded correctly.  The receiver
stores all the symbols it receives until the message is decoded
correctly.  


\paragraph{Pruning the tree.}  Decoding along the tree has exponential
complexity.  A natural greedy approximation is to prune the tree by
maintaining a small number of candidates with the lowest path costs at
each depth, while exploring the tree from root to leaves. Iteratively,
at each depth, expand the retained (up to) $B$ candidates into $B 2^k$
possible candidates at the next depth of the tree. Compute the path
cost of all of these $B 2^k$ candidates and retain the $B$ out of them
with the lowest possible path cost (break ties arbitrarily).  We use
the term {\em beam width} to refer to the parameter $B$, as this tree
exploration and pruning method is called beam
search~\cite{steinbiss1994improvements} in AI, and known as the
$M$-algorithm~\cite{anderson1984sequential} in the coding literature,
where it has been proposed for decoding convolutional codes.  We show
the somewhat surprising and noteworthy result that this simple greedy
method essentially achieves channel capacity when used for spinal
decoding.

\paragraph{Encoding and decoding complexity.}  The encoder produces
$n/k$ spine values each with $\nu$ bits.  Since the cost of producing
$\nu$ hash bits from a $\nu$-bit and $k$-bit input is $O(\nu+k)$, the
encoding cost (due to hash function calculations) scales as $O((\nu +
k)n/k) = O(n(1+ \nu/k))$.  The decoder uses the pruned tree search
over $n/k$ depth tree with each depth requiring sorting $B\cdot 2^k$
numbers as well as $B\cdot 2^k$ hash operations. Therefore, the total
decoding cost scales as $O(n B 2^k (k + \log B + \nu))$.
%
%
%
%

\medskip
\section{Performance of Spinal Codes over the BSC}
\label{ss:bscproof}

The principal result of this section is a proof of
Theorem~\ref{thm:bsc} (stated below), showing the polynomial-time
encoder and greedy tree-pruning decoder for spinal codes achieve
Shannon capacity over the BSC.

\paragraph{Model: Memoryless Channel in Discrete Time.} A noisy
channel is described by an input alphabet $\cI$, an output alphabet
$\cO$, and a collection of probability measures $\cP = (P_{i}, i \in
\cI),$ defined over $\cO$: when input $i \in \cI$ is transmitted over
the channel, the received output is distributed over $\cO$ according
to $P_i$. The communication channel is memoryless: the
output of the channel at any time depends only on the input at that
time, independent of past transmissions.  That is, when $x_1,\dots,
x_T$ are transmitted on the channel, the probability (density)
that the output is $y_1,\dots, y_T$
is $\prod_{t=1}^T P_{x_t}(y_t)$.


The BSC is memoryless. In a BSC with bit-flip probability $p \in
(0,1/2)$, $\cI = \cO = \{0,1\}$, $P_0(0) = P_{1}(1) = 1-p$, and
$P_0(1) = P_1(0) = p$. 


\begin{theorem}\label{thm:bsc}
  Consider an $n$-bit message encoded with a spinal code with $k \geq
  1$ and $\nu = \Theta(k^2 \log n)$ operating over a BSC with
  parameter $p \in (0,1/2)$.  Then, the greedy decoder with $B =
  n^{O(k^3)}$ decodes all but the last $O(k^3 \log n)$ message bits
  successfully with probability at least $1-1/n^2$,
  achieving a rate 
\begin{align}
R \geq C - O\Big(\frac{C^2}{k}\Big),  
~~\text{where } C = 1 - H(p).
\end{align}
\end{theorem}

The randomness in Theorem~\ref{thm:bsc} is induced by the channel
conditions and the code construction.  For $n = \omega(k^3 \log n)$,
the theorem says that essentially all bits are decoded (to decode {\em
  all} the bits, we can append $O(k^3 \log n)$ ``tail'' bits to the
end of each input message).  For $n \geq k^5$, the loss of rate due to
these tail bits is $O((k^4 \log n)/n) = o(1/k)$.  Therefore, the code
achieves a rate within $O(1/k)$ of the capacity of the BSC, making it
a good rateless code.  The encoder complexity scales as $O(n k \log
n)$; the decoder complexity scales as $n^{O(k^3)}$.

\paragraph{Proof plan.} The rest of this section establishes this
result with the following plan. We start by recalling Gallager's
result on the probability of error for a random code, which requires
the codewords associated with distinct messages to be completely
independent.  We present a useful variant of this result, which
requires only {\em pairwise} independence (a property, we show to be
satisfied by different enough messages under application of the hash function, see Proposition \ref{def:output_avalanche}).  We then discuss a corollary of the result
for a code operating at a rate close to the capacity, and establish
that spinal codes can operate at a rate near the capacity (so that the
corollary will apply).  Finally, we use these propositions to prove
Theorem \ref{thm:bsc} in two stages: first, assuming no hash function
collisions, and then showing that the collision probability is small.

%

\subsection{Error probability of random codes}\label{ssec:1.5}


The random code for a message of $n$ bits is constructed using a
distribution $Q$ over the input symbols. For the BSC, the input
symbols are $\{0,1\}$ and a capacity-achieving random code utilizes
$Q$ such that $Q(0) = Q(1) = 1/2$. The code maps an $n$-bit message,
$\bm \in \{0,1\}^n$, to a $T$-symbol codeword $\bx(\bm) = (x_1(\bm),
\dots, x_T(\bm))$ by drawing each of the $x_t(\bm), ~1\leq t\leq T$
independently at random according to $Q$. In the random code,
introduced by Shannon and considered by Gallager, all $\bx(\bm)$ are
independent across $\bm \in \{0,1\}^n$.  We  consider a random
code with {\em pairwise} independence across messages.  
\begin{property}[{\em Pairwise independent random code for
    the BSC}]\label{prop:bsc}
  A code that maps every $n$-bit message $\bm \in \{0,1\}^n$ to a
  random codeword of $T$ bits, $\bx(\bm)$, so that (i) for a given
  $\bm$, $x_1(\bm),\dots, x_T(\bm)$ are i.i.d. and uniformly
  distributed over $\{0,1\}$, (ii) for any $\bm \neq \bm'$,
  $\bx(\bm)$ and $\bx(\bm')$ are independent of each other, and (iii)
  the joint distribution of all codewords is symmetric.
\end{property}

For pairwise independent random codes, the following variant of
Gallager's error-exponent result \cite{Gsimple} \cite[Theorem 5.6.1,
Example 1]{Gallager68} holds (proof in
Appendix~\ref{app:variant}): 
\begin{lemma}\label{lem:bsc}
  Consider a BSC with parameter $p \in (0,1/2)$ and capacity $C = 1 -
  H(p)$.  Given a pairwise independent random code for the BSC of
  message length $n$, code length $T$, and rate $R = n/T < C$, let the
  decoder operate using the Maximum Likelihood (ML) rule to produce an
  estimate $\hat{\bm}$ when message $\bm \in \{0,1\}^n$ is
  transmitted.  Then the probability of decoding error, $P_e = 2^{-n}
  \Big(\sum_{\bm \in \{0,1\}^n} \P(\bm \neq \hat{\bm})\Big),$ for $R =
  1 - H(q)$ with $p < q$ satisfies:
\begin{itemize}
\item[(a)] $P_e  \leq 2^{-T D(q \| p)}, $ where $D(q \| p) = q \log
  \frac{q}{p} + (1-q) \log \frac{1-q}{1-p}$, if $q \leq
  \sqrt{p}/(\sqrt{p} + \sqrt{1-p})$, and
\item[(b)] $P_e \leq 2^{-T (1-R - 2 \log (\sqrt{p} + \sqrt{1-p}))}$,
  otherwise. 
\end{itemize}
\end{lemma}

\subsection{Error probability at rates $R$ close to capacity $C$}

\begin{lemma}\label{cor:bsc}
Consider the same setup as Lemma \ref{lem:bsc} with rate 
$R = n/T = 1 - H(q)$ close to capacity $C = 1 - H(p)$, that is, $q \approx p$.  Then,
\begin{align}
P_e & \leq 2^{-T D(q \| p)} ~\approx~2^{-T \kappa_p (C-R)^2}, 
\end{align}
where $\kappa_p^{-1} = \Theta\big(p (1-p) \big(\log \frac{(1-p)}{p}\big)^2\big)$. 
\end{lemma}

\begin{proof}
  From Lemma \ref{lem:bsc}, for all $q$ close enough to $p$, $P_e \leq 2^{-T D(q \| p)}$. Now, consider
  $p$ fixed and let $F(q) = D(q \| p)$ be function of $q$. Then, by
  Taylor's expansion of $F(q)$ around $p$,
\begin{align}
F(q) & = F(p) + F'(p) (q-p) + F''(\theta) (q-p)^2/2, 
\end{align}
for $\theta \in [p, q]$. Noting that
$F'(x)  = \log \Big(\frac{x (1-p)}{(1-x) p}\Big)$ and
$F''(x)  = \frac{1}{x(1-x) \ln 2},$
we see that $F(p) = F'(p) = 0$, and that for $q \approx p$,
\begin{align}\label{eq:cor.1}
F(q) & \approx \frac{(q-p)^2}{p (1-p) \ln 4}.
\end{align}
For the entropy function $H(x) = - x \log x - (1-x) \log (1-x)$, using
the first-order Taylor expansion, we obtain that for $q \approx p$,
\begin{align}
H(q) & \approx H(p) + \log \Big(\frac{1-p}{p}\Big) (q-p). 
\end{align} 
Since $R = 1 - H(q)$ and $C = 1-H(p)$,
\begin{align}\label{eq:cor.2}
q-p & \approx \frac{(C-R)}{\log \big(\frac{1-p}{p}\big)}. 
\end{align}
The desired claim follows from \eqref{eq:cor.1} and \eqref{eq:cor.2}.
\end{proof}

\subsection{Rates achievable by spinal codes}
 
The following claim shows that a spinal code over the BSC can achieve
rates arbitrarily close to the channel capacity, $C$, for large
$k$. Hence, Lemma~\ref{cor:bsc} is applicable.

\begin{claim}\label{claim:close_to_capacity}
There exists $L \geq 1$ so that the rate induced by the spinal code
at the end of pass $L$  satisfies
$$C - R  = \Theta\Big(\frac{C^2}{k}\Big). $$
\end{claim}

\begin{proof}
Consider $L$ such that $\frac{k}{L-2} \geq C > \frac{k}{L-1}$.  
These conditions may be rewritten as
\begin{align}
\frac{k}{C} + 1 &< L ~\leq \frac{k}{C} + 2 &
\frac{C}{L} & < C - \frac{k}{L} ~\leq \frac{2C}{L}. \nonumber
\end{align}  
Hence, $L=\Theta(\frac{k}{C})$, and $C - R = \Theta(\frac{C}{L})$. Together $C - R = \Theta(\frac{C^2}{k})$. 
\end{proof}

\subsection{Proof of Theorem \ref{thm:bsc}}

We now establish that by the end of pass $L$, chosen as above, decoding happens with
high probability. We shall prove that if $B=n^{O(k^3)}$, with high probability, for $i^* = \Theta\big(k^2 C^{-4} \kappa_p^{-1} \log
n\big)$, when processing the \ith~ spine value, all non-pruned
codewords either agree with the $i - i^*$ true spine values (so there
are
less than $B$ of them), or are less likely than the true spine (so
cannot cause the true spine to be pruned out). As a consequence, the
true spine is never pruned, so the decoder manages to decode all but
$i^*k = \Theta\big(k^3 C^{-4} \kappa_p^{-1} \log n\big)$ bits.

The following proposition is an implication of the strong avalanche
criterion. 

\begin{proposition}\label{def:output_avalanche}
Let $\bm$, $\bm'$ be two messages differing in message block
$\bar{m}_{i}$. Let $\{s_j\}$ and $\{s_j'\}$ be the spines for $\bm$
and $\bm'$, respectively. Then, 
$$ \P(\exists j \in \{1,\cdots g\} : s_{i+j} = s_{i+j}') \leq g \cdot 2^{-\nu}$$
If such a $j$ does not exist, then all the bits of $s_{i+1} \cdots
s_{i+g}$ are independent of bits of $s_{i+1}' \cdots s_{i+g}'$, and
each of them has a uniform independent distribution.
\end{proposition}

\begin{proof}
  Due to the pairwise independence property of hash function, when two different inputs are
  passed through the hash function, the output bits corresponding to
  these input bits are independent of each other and each of them is
  distributed independently and uniformly.  Therefore, the chance of
  two different inputs producing the same output is $2^{-\nu}$. By the
  union bound, the probability of such an event happening over a
  series of $g$ spine values is bounded by $g \cdot 2^{-\nu}$. By
  iteratively applying the property that when spine values differ at
  some stage $t$, the bits produced at stage $t+1$ are independent and
  uniformly distributed, we conclude that if all spines are different,
  their bits are independent and distributed uniformly.
\end{proof}



\paragraph{Proving Theorem \ref{thm:bsc} assuming no collisions.} We
establish Theorem \ref{thm:bsc} assuming no hash function
collisions. Later we show that collisions happen with low
probability. We require $\nu = \Theta(k^2 \log n)$ with a large-enough
constant multiplier in $\Theta(\cdot)$ term.   
Throughout, we will assume that this $\bm \in \{0,1\}^n$ is a fixed choice
that was transmitted. Establishing that with high probability (with respect to
all randomness in code construction and channel noise) it gets decoded
will imply all messages get decoded with high probability due to 
symmetry of the random-code and memoryless property of the 
BSC noise model (or more generally, any memoryless channel).

\begin{lemma}
  Consider the greedy spinal decoder operating after all coded bits of
  the $L$ passes are received.  Assuming no hash collisions, the
  decoder decodes all but the last $O(k^3 \log n)$ bits correctly with
  probability $1-O(1/n^4)$.
\end{lemma}


\begin{proof}
  Consider message $\bm$ that was transmitted and any other message $\bm'$ that differs
  from $\bm$ in any of the first $k$ bits. In the absence of hash
  collisions, as per Lemma \ref{def:output_avalanche}, 
  codewords of $\bm$ and $\bm'$ are independent of each other and each
  of their bits is independent and uniformly distributed over
  $\{0,1\}$.  That is, $\bm$ and $\bm'$ satisfy Property
  \ref{prop:bsc}.

  If we restrict our attention to codewords generated from the first
  $i$ spine values, that is, codewords of length $N = i L$, there are
  $2^{ik-k}$ codewords, one each for a message $\bm'$ that differs
  from $\bm$ in any of the first $k$ bits.  As established above, the
  pair $\bm$ and any other $\bm'$ satisfies Property \ref{prop:bsc}.
  Using Lemmas \ref{lem:bsc} and \ref{cor:bsc}, we obtain that the
  probability that any of the $2^{ik-k}$ messages (that differ from
  $\bm$ in any of the first $k$ bits) is more likely than the original
  message $\bm$ is bounded above by $P_e(i)$, where (with $k$ large
  enough for Lemma \ref{cor:bsc} to be applicable)
\begin{align}\label{eq:bsc-error}
P_e(i) & = 2^{- N \kappa_p (C-R)^2} 
         ~= 2^{- i L \frac{\kappa_p C^4}{k^2}}. 
\end{align}
That is, for $i^* = \Theta\big(k^2 C^{-4} \kappa_p^{-1} \log n\big)$,
the probability of such an error is bounded above by $1 - 1/n^6$ (with
a suitably large constant factor in the $\Theta(\cdot)$ term for $i^*$).
Therefore, after processing
the first $i^*$ spines, the only messages that can have a higher
likelihood than the original message are those that do not differ from
$\bm$ in the first $k$ bits.  There are at most $2^{i^*k-k}$ such
messages and hence if $B = 2^{i^* k} = n^{O(k^3)}$, then the original
message will not be pruned out.

Now we apply the above argument inductively. Consider a stage $j$
where the only messages that are not pruned out and have likelihood
higher than the original message $\bm$ are those that differ from
$\bm$ in a bit position between $jk - i^* k+1$ to $jk$. Now when the
decoder moves to stage $j + 1$, messages that are not pruned out are
expanded by factor $2^k$. Among these, consider the messages that
start differing from the original message in any of the $k$ bit
positions: $jk -i^*k +1, \dots, jk - i^*k + k$. By applying the same
argument as we did above, it follows that at the end of stage $j +
1$, all of these $2^{i^*k - k}$ messages will have likelihood smaller
than the original message with probability at least $1-1/n^6$.

The above invariant together with the union bound implies that at the
end of stage $n/k$, the original message is preserved in the $B$
candidates with probability at least $1-O(1/n^5)$. Further, the most
likely $2^{i^* k}$ of these $B$ candidates are those that have correct
$n - i^*k$ prefix bits.  That is, the decoder manages to decode all
but last $O(i^*k) = O(k^3 \log n)$ bits correctly.
\end{proof}

\paragraph{Dealing with collisions.} 
The above proof uses the fact that in the absence of
collisions, given the original message $\bm$ of interest and any other
message $\bm'$ that differs from $\bm$ in the first $k$ bits, their
corresponding codewords $\bx = \bx(\bm)$ and $\bx' = \bx'(\bm')$
satisfy Property \ref{prop:bsc}.  Therefore, the probability of any
such message $\bm'$ having likelihood higher than $\bm$ is at most
$O(1/n^6)$ as desired.  Note that this is precisely the argument that
is used inductively along with the union bound to establish the claim.
Therefore, it is sufficient to establish that the effect of collision is
negligible for this step only.

We wish to show that the effect of collisions is small, using the
following plan. As stated in Lemma \ref{lem:col-1}, we will identify
an event $\cE$ so that conditioned on it happening, Property
\ref{prop:bsc} is satisfied as above; and, the probability of event
$\cE^c$ is $O(1/n^6)$.  Using this, we will establish that the
probability of any such message $\bm'$ having likelihood higher than
$\bm$ continues to remain at most $O(1/n^6)$ as desired.

\begin{lemma}\label{lem:col-1}
  Let $\bm$ be the $i^*k$ prefix bits of an uncoded message. 
 Consider any other message prefix $\bm'$ of the same length,
with any of the first $k$
  bits differing from $\bm$. Then there exists event $\cE$ so that
\begin{itemize}
\item[(a)] Conditioned on event $\cE$, all pairs of messages $(\bm,
  \bm')$, satisfy Property \ref{prop:bsc}.

\item[(b)] The probability of $\cE^c$ is $O(1/n^6)$. 
\end{itemize}  
\end{lemma}

The proof of Lemma \ref{lem:col-1} is in
Appendix~\ref{app:collision}. Using the above propositions, we
complete the proof of
Theorem \ref{thm:bsc} here.  Define the event {\sf err} as the one in
which
the likelihood of an undesirable message (prefix) $\bm'$ is higher
than original message (prefix) $\bm$. Conditioned on event $\cE$, as
per Lemma \ref{lem:col-1}(a), Property \ref{prop:bsc} is satisfied by
all relevant codeword pairs as desired in the proof of Theorem
\ref{thm:bsc} in the absence of collisions. Therefore, conditioned on
event $\cE$, and the arguments presented earlier for the no-collision case,
it follows that
$\P\big({\sf err} | \cE \big)  = O\big(1/n^6\big). $
That, together with Lemma \ref{lem:col-1}(b), yields 
\begin{align}
\P\big({\sf err} \big) & = \P\big({\sf err} \cap \cE \big) + \P\big({\sf err} \cap \cE^c \big) \nonumber
 \leq \P\big({\sf err} | \cE \big) \P\big(\cE\big) +  \P\big(\cE^c\big) \nonumber 
 \leq \P\big({\sf err} | \cE \big) + O\big(1/n^6\big) \nonumber 
 = O\big(1/n^6\big). 
\end{align} 
This completes the proof of Theorem \ref{thm:bsc}.

\section{Performance of Spinal Codes over AWGN}
\label{s:awgn}

The main result of this section is that spinal codes achieve Shannon
capacity over the AWGN channel with a polynomial-time encoder and
decoder.  The arguments are similar to the BSC case.

\paragraph{AWGN channel model.}
The transmitter's primary resource is power, measured as the squared
value of the output symbols.  Typically, for regulatory and practical
reasons, the average power should be $\le P$ for some $P$.  If an
$n$-bit message $\bm = (m_1,\dots, m_n)$ is mapped to $T$ symbols
$\bx(\bm) = (x_1(\bm),\dots, x_T(\bm))$, then the power of $\bx(\bm)$
is $\frac1T\sum_i x_i^2(\bm)$.  The rate of such a code is $R = n/T$
bits/symbol. When these symbols are transmitted over the AWGN channel,
the receiver sees $\by = \bx + \bz$, where noise-vector $\bz =
(z_1,\dots, z_T)$ has i.i.d. Gaussian components with mean $0$ and
variance $\sigma^2$. The capacity of this channel is
$\CAP(P)  = \frac{1}{2} \log_2 \left(1+ \snr\right)
~\text{bits/symbol}$, 
where $\snr = \frac{P}{\sigma^2}$ denotes the signal-to-noise ratio. 

\paragraph{Encoder and Decoder.}
The procedure described in \S\ref{ss:bscenc} for generating output
symbols for the BSC is modified slightly to produce a stream of coded
symbols in $\mathbb R$. 
The modified encoder generates coded symbols from each $\nu$-bit spine
value:
in the first pass, the encoder produces $n/k$ symbols $x_1,\dots,
x_{n/k}$ using the $\MapN$ most significant bits of $s_1,\dots,
s_{n/k}$, respectively. In the next pass, the next $\MapN$
most-significant bits are used, and so on.

The sequence of $c$ input bits is treated as a binary number
$b\in\{0,\ldots,2^\MapN-1\}$. The encoder computes each output symbol
as $x_i = \Phi^{-1}\!\left(\gamma + (1 - 2\gamma)u\right)\!\sqrt{P}$,
where $\Phi$ is the CDF of standard Gaussian, $u = (b+1/2) / 2^\MapN$,
and $\gamma = \Phi(-\beta)$.  The symbols generated are in the range
$[-\beta \sqrt{P}, \beta \sqrt{P}]$, and within that range they are
distributed like a Gaussian with mean $0$ and variance $P$, quantized
into $2^\MapN$ equally-probable values. When $\beta, \MapN \to \infty$,
the coded symbols will be i.i.d. Gaussian.

The only change to the decoder is to use the squared $\ell_2$
Euclidean distance instead of the Hamming distance in
\eqref{eq:decompose}.  The intuition is that in each case, given the
channel parameters, the distance metric gives (up to normalization)
the log likelihood that a message is correct given the observation
$\mathbf y$.

\paragraph{Performance over AWGN.}
The following result shows that spinal codes achieve nearly optimal
rates over the AWGN channel in a rateless manner with efficient
encoding and decoding algorithms.

\begin{theorem}\label{thm:awgn}
Consider an AWGN channel with noise variance bounded below by
$\sigma_{\text{min}}^2$.  Consider a spinal code constrained to have
average power $\le P$, with $k > \frac{1}{2} \log\big(1+
P/\sigma^2_{\text{min}}\big)$. Let the code map $n$ message bits to
coded symbols with $\beta, \MapN$ as per \eqref{eq:awgn.7}, with
$\beps = 1/k$ and $\sigma_{\text{min}}$ in place of $\sigma$ in
\eqref{eq:awgn.7}. Let  $\nu = \Theta(k^2 \MapN \log n)$, and let the
decoder operate with $B = n^{O(k^2)}$. Then the decoder will correctly
decode all but the last $O(k^2 \log n)$ message bits with probability
at least  $1-1/n^2$ within time $T$ such that the induced rate $R =
n/T$ satisfies
\begin{align}\label{eq:awgn}
R & \geq \CAP(P) - O(1/k).
\end{align}
\end{theorem}
The proof involves a choice of parameters 
$\beta = \Theta(\sqrt{\log k})$
and $\MapN = \Theta(| \log \snr| + |\log \sigma_{\text{min}}| + \log
k)$.   These
parameters depend on $\sigma_{\text{min}}$, to bound the ``dynamic
range'' of the channel capacity.

\begin{proof}


  The highest rate at which the code can operate is $k$. Choose $k$
  large enough so that $k > \frac{1}{2} \log (1 + \snr_{\text{max}})$,
  where $\snr_{\text{max}} = P/\sigma_{\text{min}}^2$.  Now let $\beps
  = 1/k$ and assign the remaining parameters as in Claim
  \ref{claim:awgn.params} (in Appendix~\ref{app:awgn_error_prob};
  mirrors Lemma~\ref{cor:bsc}).

The proof of Claim \ref{claim:close_to_capacity} still holds with
$C = \CAP - 1/k$, so a rate $R$ such that 
$C - R  = \Theta\Big(\frac{C^2}{k}\Big)$
is achievable.
That is, $R = \CAP - \Theta(1/k)$.  
Theorem~\ref{thm:awgn} proceeds according to the same arguments as the
proof of Theorem~\ref{thm:bsc}, with Claim~\ref{claim:awgn.params}
replacing Lemma~\ref{cor:bsc}, to achieve (for large enough $k$) the
bound
\begin{align}\label{eq:awgn.10}
P_e & \leq 2^{-N (\CAP - 1/k - R)}~=~2^{- \Theta(N C^2/k)}.  
\end{align}
Subsequently, \eqref{eq:bsc-error} is replaced by
\begin{align}\label{eq:awgn-error}
P_e(i) & = 2^{- \Theta(i L C^2/k)}. 
\end{align}
That is, $i^*$ is chosen to be $\Theta\big(k C^{-2} \log n\big)$
rather than $\Theta\big(k^2 C^{-4} \kappa_p^{-1} \log n\big)$, and now
$B = 2^{i^* k} = n^{O(k^2)}$, rather than $n^{O(k^3)}$.  Finally,
$\nu$ is required to be $\Theta\big(k \log n \MapN\big)$ rather than
$\Theta\big(k^2 \log n\big)$.
\end{proof}

\section{Conclusion}\label{sec:conclusion}

We proved that spinal codes achieve Shannon capacity for the BSC and
AWGN channels with an efficient polynomial-time encoder and decoder;
they are the first rateless codes with these properties.  The key idea
in the spinal code is the application of a hash function
in a sequential manner over the message bits. The sequential structure
of the code turns out to be crucial for efficient decoding, while the
pair-wise independence of the hash function provides enough
pseudo-randomness to ensure that the code essentially achieves
capacity.

The key idea in the proof is an unusual application of a variant of
Gallager's famous result characterizing the error exponent of random
codes for any memoryless channel; the use of this result is
unconventional because the spinal code is not a traditional random
code.  Our work provides a methodical and effective way to
de-randomize Shannon's random codebook construction, and as such, applies
immediately to all discrete memoryless channels and will
likely generalize to other random coding arguments in Information
Theory.


\bibliographystyle{plain}
\bibliography{../rateless}

\begin{appendix}

\section{Proof of Lemma \ref{lem:col-1}}
\label{app:collision}

\begin{proof} To construct event $\cE$, consider the original message
(prefix) $\bm$ and any other message (prefix) $\bm'$ that differs from
$\bm$ in any of the first $k$ bits. Both these messages are of length
$i^* k$. Given $\bm$, denote all such message (prefixes) $\bm'$ as
$\cM'(\bm)$ (note that $|\cM'(\bm)| = 2^{i^*k - k}$).

At the end of $L$ passes, the codewords generated based on these
(prefix) messages are of length $N =i^* L$. Let them be $\bx$ and
$\bx'$ respectively.  We wish to evaluate the joint probability of
$\bx = \bb$ and $\bx = \bb'$ for any $\bb, \bb' \in \{0,1\}^N$
(effectively, we are assuming a re-indexing of the coded bits so that
the first $L$ coded bits depend on the first spine value, the next $L$
coded bits depends on the next spine value, and so on). Since the
messages $\bm$ and $\bm'$ differ in the first $k$ bits, by the
property of the hash function (Proposition
\ref{def:output_avalanche}), the $\nu$ bits of the first spine values
for the messages are i.i.d. uniform random bits. If the first spine
values of the two messages differ (i.e., no collision), which happens
with probability $1-2^{-\nu}$, the $\nu$ bits of the second spine
values for the two messages are i.i.d. uniform random bits, and so on.
Let $E_j$ be the event that the first $j$ spine values for both
messages are not the same (i.e., no collision amongst first $j$ spine
values). Then $\P(E_j | E_{j-1}) = 1 - 2^{-\nu}$. Therefore, for $j
\geq 1$ and since $E_{j} \subset E_{j-1}$,
\begin{align}
\P(E_j) & = \P(E_j \cap E_{j-1}) ~=~ \P(E_j | E_{j-1}) \P(E_{j-1})
\nonumber \\
           & = (1-2^{-\nu}) \P(E_{j-1}) \nonumber \\
           & = (1-2^{-\nu})^j. 
\end{align} 
Now, conditioned on $E_{j-1}$, the $L$ coded bits generated from the
\jth{} spine value in $\bx$ and $\bx'$ are i.i.d. and uniformly
distributed. Therefore, (with notation $\bx_{i,j} = (x_i, \dots,
x_j)$, etc.)
\begin{align}
\P(\bx = \bb, \bx' = \bb') & \geq \P(\bx = \bb, \bx' = \bb', E_1)
\nonumber \\
& = \P(\bx_{L+1,N} = \bb_{L+1,N}, \bx'_{L+1,N} = \bb'_{L+1,N} | E_1)
\P(\bx_{1, L} = \bb_{1, L}, \bx'_{1, L} = \bb'_{1, L}, E_1) \nonumber
\\
& \geq \P(\bx_{L+1,N} = \bb_{L+1,N}, \bx'_{L+1,N} = \bb'_{L+1,N} |
E_1) \Big(\P(\bx_{1, L} = \bb_{1, L}, \bx'_{1, L} = \bb'_{1, L}) -
\P(E_1^c)\Big) \nonumber \\
& = \P(\bx_{L+1,N} = \bb_{L+1,N}, \bx'_{L+1,N} = \bb'_{L+1,N} | E_1)
\Big(\P(\bx_{1, L} = \bb_{1, L}) \P(\bx'_{1, L} = \bb'_{1, L}) -  
\P(E_1^c)\Big) \nonumber \\
& = \P(\bx_{L+1,N} = \bb_{L+1,N}, \bx'_{L+1,N} = \bb'_{L+1,N} | E_1)
\Big( 2^{-2L} - 2^{-\nu}\Big) \nonumber \\
& = \P(\bx_{L+1,N} = \bb_{L+1,N}, \bx'_{L+1,N} = \bb'_{L+1,N} | E_1)
2^{-2L} (1-2^{-\nu + 2L}) \nonumber \\
& =  (1-2^{-\nu + 2L}) ~\P(\bx_{1, L} = \bb_{1, L}) \P(\bx'_{1, L} =
\bb'_{1, L}) ~\P(\bx_{L+1,N} = \bb_{L+1,N}, \bx'_{L+1,N} =
\bb'_{L+1,N} | E_1). \label{eq:rec}
\end{align}
Here, we have used the fact that the distribution of $\bx_{L+1, N},
\bx'_{L+1, N}$ is conditionally independent of $\bx_{1, L},
\bx'_{1,L}$ given $E_1$.  \eqref{eq:rec} then sets up a recursion,
leading to the following:
\begin{align}
\P(\bx = \bb, \bx' = \bb') & \geq \P(\bx = \bb) \P(\bx' = \bb')
\big(1-2^{-\nu + 2L}\big)^{i^*}. 
\end{align}
Given this, with respect to the underlying probability space,
$\Omega$, we can define an event $\cE_{\bm, \bm'}$ with
\begin{align}\label{eq:col-2}
\P(\cE_{\bm, \bm'}) & = \big(1-2^{-\nu + 2L}\big)^{i^*}, 
\end{align}
so that we have 
\begin{align}\label{eq:col-3}
\bone_{\{\bx = \bb, \bx' = \bb'\}}& =  \P(\bx = \bb) \P(\bx' = \bb')
\bone_{\{\cE_{\bm,\bm'}\}} + q(\bb, \bb')
\bone_{\{\cE_{\bm,\bm'}^c\}}, 
\end{align}
where $\bone_{\{E\}}(\cdot)$ is the indicator random variable of event
$E$ with $\bone_{\{E\}}(\omega) = 1$ if $\omega \in E$ and $0$
otherwise for $\omega \in \Omega$ and $q(\cdot, \cdot)$ represents the
conditional probability distribution of $\bx, \bx'$ given
$\cE^c$. Equivalently, what we have is an event $\cE_{\bm, \bm'}$ with
property \eqref{eq:col-2} such that
\begin{align}
\P(\bx = \bb, \bx' = \bb' | \cE_{\bm, \bm'}) & = \P(\bx = \bb) \P(\bx'
= \bb'). 
\end{align} 
Now define 
\begin{align}
\cE & = \cap_{\bm' \in \cM'(\bm)} \cE_{\bm, \bm'}. 
\end{align} 
Since $\cE \subset \cE_{\bm, \bm'}$ for any $\bm' \in \cM'(\bm)$ and
from \eqref{eq:col-3}, the conditional distribution of $\bx, \bx'$
with respect to $\cE_{\bm, \bm'}$ is uniform, it follows that, for any
$\bm' \in \cM'(\bm)$,
\begin{align}
\P(\bx = \bb, \bx' = \bb' | \cE) & = \P(\bx = \bb) \P(\bx' = \bb'). 
\end{align} 
Finally, by \eqref{eq:col-2} and union bound, it follows that 
\begin{align}
\P\big(\cE^c\big) & = O(i^* 2^{N + 2L - \nu})
\end{align}
Therefore, choosing
\begin{align}\label{eq:nu}
\nu & =  (N + 2L + \log i^*) + 6 \log n  ~=~\Theta\big( k^2 \log n
\big), 
\end{align}
with an appropriately large constantleads to
\begin{align}
\P\big(\cE^c\big) & = O\big(1/n^6\big),
\end{align}
as desired, completing the proof of Lemma \ref{lem:col-1}. 
\end{proof}

\section{Error probability of random codes over
AWGN}\label{app:awgn_error_prob}
As in \S\ref{ssec:1.5}, we consider random codes with only pairwise
independent codewords across messages.
\begin{property}[{\em Pairwise independent random code for AWGN with
distribution $Q$}]\label{prop:awgn}
A code that maps every $n$ bit message $\bm \in \{0,1\}^n$ to a
random codeword of $T$ real numbers, $\bx(\bm)$, so that (i) for a
given $\bm$, $x_1(\bm),\dots, x_T(\bm)$ are i.i.d. with distribution
$Q$, (ii) for any $\bm \neq \bm'$, $\bx(\bm)$ and $\bx(\bm')$ are
independent of each other, and (iii) the joint distribution of all
codewords is symmetric. 
\end{property}

This definition allows us to state the following variant of Gallager's
error-exponent result \cite[Theorem 7.3.2]{Gallager68} for a random
code on the AWGN channel.  The coded symbols have distribution $Q$
over a finite set $\Omega \subset \mathbb R$.

\begin{lemma}\label{lem:awgn}
Consider an AWGN channel with noise variance $\sigma^2$ and a
pairwise independent random code for AWGN with distribution $Q$,
message length $n$, code length $T$, and rate $R = n/T < C$.
Then the probability of error under ML decoding is bounded by
\begin{align}
P_e & \leq 2^{-T (E_o(Q) - R)}, \quad \text{where} \quad E_o(Q) =
-\log \Bigg\{ \frac{1}{\sqrt{2\pi \sigma^2}} \int_{\R} \Big[\sum_{j
\in \Omega} Q(j) \exp\Big(-\frac{(y - j)^2}{4 \sigma^2}\Big)\Big]^2 dy
\Bigg\}.
\end{align}
\end{lemma} 
Next, we want to specialize this bound for the spinal code symbol
distribution described in \S\ref{s:awgn}. Given
$\MapN$, $\beta$ and $P$, let
\[ 
\Omega = \Big\{\Phi^{-1}\!\!\left(\gamma + (1 -
2\gamma)u\right)\!\sqrt{P}: ~\gamma = \Phi(-\beta),~u = \frac{b +
1/2}{2^{\MapN}},~b \in \{0,\dots, 2^\MapN-1\} \Big\}. 
\]
where $\Phi$ is the CDF of the standard Gaussian.  By construction,
$|\Omega| = 2^\MapN$ and $\Omega \subset [-\beta \sqrt{P}, \beta
\sqrt{P}]$.  The distribution over $b$ is uniform, as in the case of
the spinal encoder, and hence each $j\in\Omega$ is equally likely
with probability $2^{-\MapN}$.  This leads to the following result.

\begin{lemma}\label{cor:awgn}
For the channel and code of Lemma~\ref{lem:awgn} with uniform
distribution over $\Omega$ (with parameters $\beta, \MapN, P$), the
probability of error under ML decoding is bounded above as 
\begin{align}\label{eq:awgn.6}
P_e & \leq 2^{-T (E' - R)}, \quad \text{where} \quad E' = \max_{\zeta > 1} \left(\frac{1}{2} \log \Big(1 + \frac{P}{\sigma^2
(1+1/\zeta)^2}\Big) - \frac{2 r(\beta)}{\ln 2} - \frac{(2\zeta + 1)
\Delta^2 \log e}{4 \sigma^2}\right).
\end{align}
\end{lemma}

\begin{proof}
Lemma \ref{lem:awgn} indicates that a random code generated with this
distribution would have error probability
\begin{align}\label{eq:awgn.1}
P_e & \leq 2^{-T (E - R)}, \quad \text{where} \quad E = -\log \Big\{
\frac{1}{\sqrt{2\pi \sigma^2}}\int_\R \Big[\sum_{j \in \Omega}
2^{-\MapN} \exp\Big(-\frac{(y - j)^2}{4 \sigma^2}\Big)\Big]^2 dy
\Big\}.  
\end{align}

The expression above is explicit but opaque. We can simplify it by
rewriting the summation over discrete $j\in\Omega$ as an integral over
$x\in[-\beta\sqrt{P},\beta\sqrt{P}]$ with Gaussian density, provided
that we construct a suitable function $\delta(x)$ so that
$j=x+\delta(x)$ is distributed according to $\Omega$.  Extracting
$\delta(x)$ from the resulting integrand will yield a tractable
expression.

Using the mean value theorem and the properties of the Gaussian, the
separation between two adjacent elements in $\Omega$ can be bounded
above by 
\[ 
\Delta \equiv \Delta(\beta, \MapN, P) = \frac{\beta \sqrt{P}
\exp(\beta^2/2)}{2^{\MapN-1}}. 
\]

Now consider the following thought experiment.  First, sample a
Gaussian variable with mean $0$ and variance $P$. If the outcome is
within $[-\beta\sqrt{P}, \beta \sqrt{P}]$, map it to a nearby value in
$\Omega$ so that the induced distribution over elements of $\Omega$ is
uniform (equiprobable quantization); if the outcome is not within
$[-\beta\sqrt{P}, \beta \sqrt{P}]$, reject it (truncation). The
rejection probability, $r(\beta)$, is $2(1-\Phi(\beta)) =
O(\frac{1}{\beta}\exp(-\beta^2/2))$. We can relate the quantized value $j$ to the
sampled Gaussian value $x$ by an additive discretization error
$\delta(x)=j-x$.  From these properties, it follows that the discrete
summation involving probabilities $2^{-\MapN}$ over $\Omega$ in
\eqref{eq:awgn.1} can be replaced by a Riemann integral over the
Gaussian density with mean $0$ and variance $P$, normalized by
$1/(1-r(\beta))$, and limited to the range $[-\beta \sqrt{P},
\beta\sqrt{P}]$:
\begin{align}\label{eq:awgn.2}
E & = -\log \Big\{ \frac{1}{\sqrt{8\pi^3 P^2 \sigma^2}} \int_\R
\Big[\int_{-\beta \sqrt{P}}^{\beta \sqrt{P}} (1-r(\beta))^{-1}
\exp\Big(-\frac{x^2}{2P} \Big) \exp\Big(-\frac{(y - x -
\delta(x))^2}{4 \sigma^2}\Big) dx \Big]^2 dy \Big\}, 
\end{align}
By construction, $|\delta(x)| \leq \Delta$. We can pull $\delta(x)$
out of the integrand by placing a multiplicative bound on
$\exp(-(y-x-\delta(x))^2/4\sigma^2)$ in terms of
$\exp(-(y-x)^2/4\sigma^2)$ and a small error term involving $\Delta$.
Let $\zeta > 1$ be a large constant.  Then if $|y-x| \leq \zeta
|\delta(x)|$,
\begin{align}\label{eq:awgn.3}
\exp\Big(-\frac{(y-x-\delta(x))^2}{4\sigma^2}\Big) & \leq
\exp\Big(-\frac{(y-x)^2}{4\sigma^2}\Big) \exp\Big(\frac{(2\zeta +1)
\Delta^2}{4\sigma^2}\Big).  
\end{align}
Otherwise, $|y-x| \geq \zeta |\delta(x)|$ and hence 
\begin{align}\label{eq:awgn.4}
\exp\Big(-\frac{(y-x-\delta(x))^2}{4\sigma^2}\Big) & \leq
\exp\Big(-\frac{(y-x)^2}{4\sigma^2 (1+1/\zeta)^2 }\Big). 
\end{align}
From \eqref{eq:awgn.2}-\eqref{eq:awgn.4}, it follows that (using
approximation $\log (1-x) \approx - x/\ln 2$ and treating $r(\beta)$
small or equivalently $\beta$ large), 
\begin{align}\label{eq:awgn.5}
E & \geq - \frac{2 r(\beta)}{\ln 2} - \frac{(2\zeta + 1) \Delta^2 \log
e}{4 \sigma^2} 
-\log \Big\{\frac{1}{\sqrt{8\pi^3 P^2 \sigma^2}}   \int_\R
\Big[\int_{\mathbb R}  \exp\Big(-\frac{x^2}{2P} \Big)
\exp\Big(-\frac{(y - x)^2}{4 \sigma^2 (1+1/\zeta)^2}\Big) dx \Big]^2
dy \Big\}. 
\end{align}
As established in \cite[Eq (7.4.21)]{Gallager68}, 
\begin{align}\label{eq:awgn.5a}
& -\log \Big\{\frac{1}{\sqrt{8\pi^3 P^2 \sigma^2 (1-1/\zeta)^2}}  
\int_\R \Big[\int_{\mathbb R}  \exp\Big(-\frac{x^2}{2P} \Big)
\exp\Big(-\frac{(y - x)^2}{4 \sigma^2 (1+1/\zeta)^2}\Big) dx \Big]^2
dy \Big\} \nonumber \\
& \qquad = \frac{1}{2} \log \Big(1 + \frac{P}{\sigma^2
(1+1/\zeta)^2}\Big). 
\end{align} 
Combining \eqref{eq:awgn.5} and \eqref{eq:awgn.5a}, we obtain the
desired result. 
\end{proof}

\begin{claim}\label{claim:awgn.params}
With an appropriate choice of parameters, for a pairwise independent
random code over AWGN with distribution $Q$, the probability of error
for a rate $R < \CAP - \beps$ is bounded as
\begin{align}\label{eq:awgn.9}
P_e & \leq 2^{-T (\CAP - \beps - R)}
\end{align}
\end{claim}

\begin{proof}
For a given small enough $\beps > 0$, select 
\begin{align}\label{eq:awgn.7}
& \zeta  = 9\snr/\beps, \quad \text{with} \quad \snr = P/\sigma^2,
\nonumber \\
& \beta \quad \text{large enough so that} \quad r(\beta) = \frac{\beps
\ln 2}{6}\quad \text{where recall}\quad r(\beta) = 2(1-\Phi(\beta)),
\nonumber \\
& \MapN \quad \text{large enough so that} \quad \Delta = \frac{2\beps
\sigma^2}{9 \sqrt{P}} \quad \text{where recall} \quad \Delta =
\frac{\beta \sqrt{P} \exp(\beta^2/2)}{2^{\MapN-1}}.
\end{align}
This selection leads to $\beta = \Theta(\sqrt{\log 1/\beps})$ and
$\MapN = \Theta(\log \snr_{\text{max}} + \log \sigma_{\text{min}} +
\log 1/\beps )$ with $\snr_{\text{max}} = P/\sigma_{\text{min}}^2$.
Now, with these choices of parameters and using the fact that
$\frac{1}{2}\log (1+x)$ is a $1$-Lipschitz function, we obtain from
\eqref{eq:awgn.6} that  
\begin{align}\label{eq:awgn.8}
E' & \geq \CAP - \beps. 
\end{align}
\end{proof}

\section{Variation of Gallager's result: Pairwise independent random
code and discrete memoryless channel} 
\label{app:variant}

Here we present a derivation of a variation of Gallager's result about
the error exponent (or error probability) for a random code under the
ML decoding rule for any discrete memoryless channel. The variation
assumes the pairwise independence property of random codewords rather
than complete independence. Effectively, we observe that the proof
technique of Gallager \cite{Gsimple, Gallager68} requires only
pairwise independence.  Since results identical to Lemmas
\ref{lem:bsc} and \ref{lem:awgn} were derived by specializing them for
the BSC and AWGN channel respectively (see \cite[Chapters 5,
7]{Gallager68}), the justification of these two Lemmas follow.


\paragraph{Pairwise independent random code.} Consider $n$-bit
messages in $\{0,1\}^n$. Let $Q$ be distribution over $\cI$. Under
a pairwise random code, using $Q$, of rate $R = n/T$, each message $\bm
\in \{0,1\}^n$ is mapped to a random codeword $\bx(\bm) \in \cI^T$
such that
\newcommand{\bi}{\mathbf i}
\begin{itemize}
\item[(a)] For any $\bm \in \{0,1\}^n$ and $\bi = (i_1,\dots, i_T) \in \cI^T$, 
\begin{align}
\P(\bx(\bm) = \bi) & = \prod_{t=1}^T Q(i_t).  
\end{align}
\item[(b)] For any $\bm \neq \bm' \in \{0,1\}^n$ and $\bi, \bi' \in \cI^T$, 
\begin{align}
\P(\bx(\bm) = \bi, \bx(\bm') = \bi') & = \P(\bx(\bm) = \bi) \times \P(\bx(\bm') = \bi')
\end{align}

\item[(c)] The joint distribution of all codewords is symmetric. 

\end{itemize}

\paragraph{Maximum likelihood decoding.} To transmit message $\bm$,
the codeword $\bx(\bm)$ is sent over the channel, producing output
$\by$.
The ML rule produces an estimate $\hat{\bm}$ so that
\begin{align}\label{apx:ml}
\P(\by | \bx(\hat{\bm})) & = \max_{\bm' \in \{0,1\}^n} \P(\by | \bx(\bm')). 
\end{align}
A decoding error occurs if $\hat{\bm} \neq \bm$.

\paragraph{Probability of error.} Let $P_{e\bm}$ denote the
probability of decoding error when $\bm$ was transmitted.  $P_{e\bm}$
is average of probability of error over all randomly chosen codes. As
before, the overall probability of error is
\begin{align}
P_e & = \frac{1}{2^n} \Big(\sum_{\bm \in \{0,1\}^n} P_{e\bm}\Big).
\end{align}
Due to symmetry in the random code, $P_{e\bm}$, the average
probability of error over all choices of codes, is the same for all
$\bm$. Therefore, $P_e$ equals $P_{e\bm}$ for any given $\bm$.

\begin{theorem}\label{thm:gallager}
Given the above setup, for any $\bm \in \{0,1\}^n$, 
\begin{align}
P_{e \bm} & \leq 2^{-T (-\rho R + E_o(\rho, Q))}, 
\end{align}
for any $0 < \rho \leq 1$ with 
\begin{align}
E_o(\rho, Q) & = - \log \Big[\sum_{j \in \cO} \Big(\sum_{i \in \cI} Q(i) P_{ij}^{\frac{1}{1+\rho}}\Big)^{1+\rho}\Big].
\end{align}
\end{theorem}
The best bound is achieved by optimizing for choice of $\rho,
Q$. Specifically, define  
\begin{align}\label{apx:eq1}
E_o(R) & = \max_{\rho, Q} \Big[ -\rho R + E_o(\rho, Q)\Big].
\end{align}
Then, Theorem \ref{thm:gallager} implies the bound $P_{e\bm} \leq
2^{-N E_o(R)}$. This bound when specialized to the BSC and AWGN
channel (with proper choice of $\rho, Q$ in \eqref{apx:eq1}) results
in Lemmas \ref{lem:bsc} and \ref{lem:awgn} (see \cite[Chapters 5,
7]{Gallager68} for details).
\begin{proof}[Proof of Theorem \ref{thm:gallager}]
  The proof is essentially identical to that in \cite{Gsimple},
  presented here for completeness. Consider a message $\bm \in
  \{0,1\}^n$. Then,
\begin{align}\label{apx:eq2}
\P(\bm \neq \hat{\bm}) & = \sum_{\by \in \cO^T} \P(\by | \bx(\bm)) \phi_\bm(\by), 
\end{align}
where
\begin{align}
\phi_{\bm}(\by) & = \begin{cases} 1 & \text{~if~} \P(\by | \bx(\bm)) \leq \P(\by | \bx(\bm')) ~~\text{for some} ~~\bm' \neq \bm, \\
                                     0 & \text{~otherwise}.
                                     \end{cases}
\end{align}
This can be upper bounded as 
\begin{align}\label{apx:eq3}
\phi_{\bm}(\by) & \leq \left[ \frac{\sum_{\bm' \neq \bm} \P(\by | \bx(\bm'))^{\frac{1}{1+\rho}}}{\P(\by | \bx(\bm))^{\frac{1}{1+\rho}}}\right]^\rho, ~\rho > 0. 
\end{align}
From \eqref{apx:eq2} and \eqref{apx:eq3}, we obtain 
\begin{align}\label{apx:eq4}
\P(\bm \neq \hat{\bm})& \leq \sum_{\by \in \cO^T} \P(\by | \bx(\bm))^{\frac{1}{1+\rho}} \Big[\sum_{\bm' \neq \bm} \P(\by | \bx(\bm'))^{\frac{1}{1+\rho}}\Big]^\rho, ~\rho > 0. 
\end{align}
Now recalling that it is a pairwise independent random code and
averaging both sides with respect to this random code, we obtain for
$0 < \rho \leq 1$, 
\begin{align}\label{apx:eq5}
P_{e\bm} & \equiv \E\big[\P(\bm \neq \hat{\bm})\big] \nonumber \\
                 & \leq \E\Big[ \sum_{\by \in \cO^T} \P(\by | \bx(\bm))^{\frac{1}{1+\rho}} \Big[\sum_{\bm' \neq \bm} \P(\by | \bx(\bm'))^{\frac{1}{1+\rho}}\Big]^\rho \Big] \nonumber \\
                 & = \sum_{\by \in \cO^T} \E\Big[  \P(\by | \bx(\bm))^{\frac{1}{1+\rho}} \Big[\sum_{\bm' \neq \bm} \P(\by | \bx(\bm'))^{\frac{1}{1+\rho}}\Big]^\rho \Big] \nonumber \\
                 & = \sum_{\by \in \cO^T} \E_{\bx(\bm)}\Big[  \P(\by | \bx(\bm))^{\frac{1}{1+\rho}} \E \Big[\Big[\sum_{\bm' \neq \bm} \P(\by | \bx(\bm'))^{\frac{1}{1+\rho}}\Big]^\rho \big | \bx(\bm)\Big]\Big] \Big] \nonumber\\
                 & \stackrel{(a)}{\leq} \sum_{\by \in \cO^T} \E_{\bx(\bm)}\Big[  \P(\by | \bx(\bm))^{\frac{1}{1+\rho}} ~\Big(\E \Big[\Big[\sum_{\bm' \neq \bm} \P(\by | \bx(\bm'))^{\frac{1}{1+\rho}}\Big] \big | \bx(\bm)\Big]\Big)^\rho \Big] \Big]  \nonumber \\
               & = \sum_{\by \in \cO^T} \E_{\bx(\bm)}\Big[  \P(\by | \bx(\bm))^{\frac{1}{1+\rho}} ~\Big( \sum_{\bm' \neq \bm} \E\Big[\P(\by | \bx(\bm'))^{\frac{1}{1+\rho}} \big | \bx(\bm)\Big]\Big] \Big)^\rho \Big] \nonumber \\ 
               & \stackrel{(b)}{=} \sum_{\by \in \cO^T} \E_{\bx(\bm)}\Big[  \P(\by | \bx(\bm))^{\frac{1}{1+\rho}}\Big] ~\Big( \sum_{\bm' \neq \bm} \E\Big[\P(\by | \bx(\bm'))^{\frac{1}{1+\rho}}\Big]\Big] \Big)^\rho. 
\end{align}
Here, we use the notation $\E_{\bx(\bm)}$ to explicitly note that the
randomness is with respect to $\bx(\bm)$; (a) follows from Jensen's
inequality for conditional expectation and fact that $f(x) = x^\rho$
is a concave function for $0 < \rho \leq 1$; and (b) follows from the
pairwise independence of $\bx(\bm)$ and $\bx(\bm')$ for any pair of
messages $\bm \neq \bm'$. Now due to symmetry of the random coding
distribution, it follows that $\E\Big[\P(\by |
\bx(\bm'))^{\frac{1}{1+\rho}}\Big]$ is the same for all $\bm'$
(including $\bm$) and equals
\begin{align}
\E\Big[\P(\by | \bx(\bm'))^{\frac{1}{1+\rho}}\Big] & = \sum_{\bi \in \cI^T} Q^T(\bi) \P(\by | \bi)^{\frac{1}{1+\rho}}, 
\end{align} 
where $Q^T(\bi) = \prod_{t=1}^T Q(i_t)$. Therefore, from
\eqref{apx:eq5} and the fact that $n = RT$, we have  
\begin{align}\label{apx:eq7}
P_{e\bm} & \leq 2^{\rho RT}~\sum_{\by \in \cO^T} \Big[\sum_{\bi \in \cI^T} Q^T(\bi) \P(\by | \bi)^{\frac{1}{1+\rho}}\Big]^{1+\rho}. 
\end{align} 
Now using the property of memoryless channels and random codes, we
have that
\begin{align}
Q^T(\bi) \P(\by | \bi)^{\frac{1}{1+\rho}} & = \prod_{t=1}^T Q(i_t) \P(y_t | i_t)^{\frac{1}{1+\rho}}.
\end{align}
Using this product-from in \eqref{apx:eq7} and exchanging sums and
products, we have
\begin{align}\label{apx:eq8}
P_{e\bm} & \leq 2^{\rho RT} \prod_{t=1}^T \sum_{y_t \in \cO} \Big[ \sum_{i_t \in \cI} Q(i_t) \P(y_t | i_t)^{\frac{1}{1+\rho}}\Big]^{1+\rho} \nonumber \\
                 & \stackrel{(a)}{=} 2^{\rho RT} \Big[ \sum_{j \in \cO} \Big(\sum_{i \in \cI} Q(i) P_{ij}^{\frac{1}{1+\rho}}\Big)^{1+\rho}\Big] \nonumber \\
                 & \stackrel{(b)}{=} 2^{\rho R T} 2^{-T E_o(\rho, Q)} \nonumber \\
                 & = 2^{ -T (- \rho R + E_o(\rho, Q))}. 
\end{align}
Here, (a) uses the definitions of random code and memoryless channel,
and (b) follows from the definition of $E_o(\rho, Q)$. 
\end{proof}

\end{appendix}

\end{document}